\documentclass{jpp}
\usepackage[square,sort,comma,numbers]{natbib}
\usepackage{graphicx}
\usepackage[utf8]{inputenc}
\usepackage[T1]{fontenc}
\usepackage{amssymb}
\usepackage{amsmath}
\usepackage{hyperref}
\usepackage{cleveref}
\usepackage{lipsum}
\graphicspath{{./fig/}{./png/}}

\newtheorem{example}{Example}
\newtheorem{remark}{Remark}

\def\K{{\bf K}}

\def\I{\bf{I}}

\def\Z{\bf{Z}}

\def\1{{\bf 1}}

\def\A{{\bf A}}

\def\k{{\bf k}}
\def\J{{\bf J}}
\def\B{{\bf B}}

\def\V{\vec{\bf V}}

\def\i{{\bf i}}
\def\j{{\bf j}}

\def\v{{\bf v}}
\def\x{{\bf x}}
\def\y{{\bf y}}

\def\grad{{\bf{grad}}}
\def\div{{\bf{div}}}
\def\rot{{\bf{rot}}}
\def\vol{{\rm vol}}

\newcommand{\dd}{{\rm d} {}}
\newcommand{\Fig}[1]{Figure~\ref{#1}}
\def\pr{{\rm pr}}
\usepackage{color}

\newtheorem{theorem}{Theorem}
\newtheorem{lemma}{Lemma}

\title{Minimum Quadratic Helicity States}

\author{
P. M. Akhmet'ev\aff{1, 2} \corresp{\email{pmakhmet@mail.ru}},
S. Candelaresi\aff{3} \and
A. Y. Smirnov\aff{2, 4}
}

\affiliation{
\aff{1}MIEM HSE, Moscow, Russia
\aff{2}IZMIRAN, Troitsk, Moscow region, Russia
\aff{3}Division of Mathematics, University of Dundee, Dundee, DD1 4HN, UK
\aff{4}National University of Science and Technology MISiS, Moscow, 119049, Russia
}

\begin{document}

\maketitle

\begin{abstract}
Building on previous results on the quadratic helicity in magnetohydrodynamics (MHD)
we investigate particular minimum helicity states.
Those are eigenfunctions of the curl operator and are shown to constitute
solutions of the quasi-stationary incompressible ideal MHD equations.
We then show that these states have indeed minimum quadratic helicity.
\end{abstract}

\section{Introduction}

Magnetic field line topology has been recognized to be a crucial part in the
evolution of magnetic fields in magnetohydrodynamics (MHD)
\cite{Woltjer-1958-489-91-PNAS, Parker-1972-174-499-ApJ, Taylor-1974-PrlE,
Frisch-Pouquet-Leorat-1975-JFluidMech, KleeorinRuzmaikin82,
Hornig-Schindler-1996-3-781-PhysPlasm,
fluxRings10, Yeates_Topology_2010, Wilmot-Smith-Pontin-2010-516-A5-AA, knotsDecay11}.
The most used quantifier of the field's topology is the magnetic helicity
\cite{MoffattKnottedness1969, ArnoldHopf1974, BergerField1984JFM, Enciso-Peralta-2016-113-8-PNAS}
which measures the linking, braiding and twisting of the field lines.
Through Arnold's inequality \cite{ArnoldHopf1974} it imposes a lower bound for the magnetic energy.
As the magnetic helicity is a (second order) invariant under non-dissipative evolution (non-resistive)
it imposes restrictions on the evolution of the magnetic field.
A further topological invariant $M$ of topological complexity $7$ can be found \citep{A2}
(the idea of the construction is presented in \cite{Ruzmaikin-Akhmetiev-1994-331-1-PhysPlasm}).
$M$ is a generalized helicity integral and constitutes a more effective lower bound for
magnetic energy compared to the magnetic helicity.
Informally, $M$ is a measure of how much the magnetic lines are of
the shape of helical Borromean rings.
Second order invariants are
the field line helicity
\citep{Yeates-Hornig-2011-18-102118-PhysPlasm, R-Y-H-W} that measures a weighted average
helicity along magnetic field lines, and the two quadratic helicities $\chi^{(2)}$, $\chi^{[2]}$,
which are to be considered as the $L^2$-norms of field line helicity.
The main problem to apply the high-order helicity is related with its calculation.
A local formula for quadratic helicity $\chi^{(2)}$ is proposed by \cite{A-C-S}.

In this work we consider 
another approach to calculate
the quadratic helicities of special cases of magnetic fields, 
which is based on the ergodic theorem.
Those are eigenvectors of the curl operator, which implies that the field is
also force-free, i.e.\ the Lorentz force vanishes.
We first introduce these fields and discuss some general properties by applying
the Lobachevskii geometry to MHD.
Then we show that they constitute quasi-stationary solutions of the ideal
incompressible MHD equations by using geodesic flows \cite{1}.
This is done on special manifolds equipped with a prescribed Riemannian metric, which corresponds to a dynamics
of the Anosov type.
Using the geodesic flow construction, we apply the results from hyperbolic dynamics
to calculate higher invariants of the magnetic field of which
presented calculations of quadratic helicities are the simplest examples.
Finally, we show that those fields constitute minimal quadratic helicity states.

\section{Eigenfunctions of the curl operator}
\label{sec: eigenfunctions}

\subsection{Positive Eigenfunction}

We recall formula for Hopf magnetic field, a generalization is presented in \cite{4}.
Let $S^3$ be the standard $3$-sphere
\begin{equation}
S^3 = \{ z_1,z_2 \vert z_1 \bar{z_1} + z_2 \bar{z_2} = 1 \}, \quad z_1, z_2 \in \mathbb{C},
\end{equation}
equipped with the standard Riemannian metric $g$. 
Let $\Theta: S^1 \times S^3 \to S^3$ be the standard action of the unit complex circle, given by
$$ \Theta(\varphi;z_1,z_2) = (z_1 \exp(\i \varphi), z_2 \exp(\i \varphi)). $$
Let $\B_{\rm right} = \dd\Theta/\dd\varphi$ be the Hopf magnetic field on $S^3$,
which is tangent to the Hopf fibers (fibers of $\Theta$).

\begin{lemma}\label{l1}
Consider the operator $\rot$ on the Riemannian manifold $(S^3, g)$ (see for the definition \cite{A-Kh} I.9.5), we get:
\begin{equation}\label{eq1}
\rot\B_{\rm right}(\x) = 2\B_{\rm right}(\x), \quad \x \in S^3.
\end{equation}
\end{lemma}

\begin{proof}
This is Example 5.2 in  \cite{ArnoldHopf1974}
However, we show here direct calculations of this lemma.
For that we define the curve $\Theta$ on $\mathbb{R}^4$ rather than on $\mathbb{C}^2$:
\begin{eqnarray}
\Theta(\varphi, x_0, x_1, x_2, x_3) & = & \left( x_0\cos(\varphi) - x_1\sin(\varphi), x_0\sin(\varphi) + x_1\cos(\varphi), \right. \nonumber \\
 & & \left. x_2\cos(\varphi) - x_3\sin(\varphi), x_2\sin(\varphi) + x_3\cos(\varphi) \right),
\end{eqnarray}
with the coordinates $x_0$, $x_1$, $x_2$ and $x_3$.
From that we can compute $\B_{\rm right} = \dd\Theta/\dd\varphi$ from which we define
the associated differential one-form on $\mathbb{R}^4$:
\begin{equation}
\beta_{\rm right}^{\rm R4} =
B^0_{\rm right} \dd x^0 + B^1_{\rm right} \dd x^1 + B^2_{\rm right} \dd x^2 + B^3_{\rm right} \dd x^3.
\end{equation}

We now define the mapping between points on the three-sphere $S^3$ and $\mathbb{R}^4$:
\begin{eqnarray}
\Psi & = & (x_0, x_1, x_2, x_3) \\
x_0 & = & \cos(\theta_1) \nonumber \\
x_1 & = & \sin(\theta_1)\cos(\theta_2) \nonumber \\
x_2 & = & \sin(\theta_1)\sin(\theta_2) \cos(\theta_3) \nonumber \\
x_3 & = & \sin(\theta_1)\sin(\theta_2) \sin(\theta_3), \nonumber
\end{eqnarray}
with the coordinates of $S^3$: $\theta_1 \in [0, 2\pi)$, $\theta_2 \in [0, \pi]$
and $\theta_3 \in [0, \pi]$.
We can now compute the differential one-form $\beta_{\rm right}^{\rm R4}$ on $S^3$ as the pull-back
under the mapping $\Psi$
\begin{eqnarray}
\beta_{\rm right}^{\rm S3} & = & \Psi^{*} \beta_{\rm right}^{\rm R4} \\
 & = & \cos(\varphi)\cos(\theta_2) \dd\theta^1 - \cos(\varphi)\cos(\theta_1)\sin(\theta_1)\sin(\theta_2) \dd\theta^2 \nonumber \\
  & & + \cos(\varphi)\sin^2(\theta_1)\sin^2(\theta_2) \dd\theta^3 \nonumber
\end{eqnarray}

The curl operation on the vector field $\B_{\rm right}^{S3}$ corresponds to the exterior differential
of the one-form $\beta_{\rm right}^{S3}$ which results in a two-form $\dd\beta_{\rm right}^{S3}$.
We take it's Hodge-dual $\star\dd\beta_{\rm right}^{S3}$,
with the volume element $\dd V = \sin^2{\theta_1} \sin{\theta_2} \dd\theta_1 \wedge \dd\theta_2 \wedge \dd\theta_3$,
compare it with $\beta_{\rm right}^{S3}$ and find 
\begin{eqnarray}
\star\dd\beta_{\rm right}^{S3} & = &
2\cos(\varphi)\cos(\theta_2) \dd\theta^1
- 2\cos(\varphi)\cos(\theta_1)\sin(\theta_1)\sin(\theta_2)  \dd\theta^2 \nonumber \\
 & & + 2\cos(\varphi)\sin^2(\theta_1)\sin^2(\theta_2)  \dd\theta^3.
\end{eqnarray}
Hence the result
\begin{equation}\label{analog}
\star\dd\beta_{\rm right}^{\rm S3} = 2\beta_{\rm right}^{\rm S3},
\end{equation}
which corresponds to equation \eqref{eq1}.

The left transformation of $S^3$ (see the beginning of the next section for the right transformation)
is transitive and is an isometry
This isometry commutes with the curl operator and keeps the Hopf fibration (which is determined by the right $\i$-multiplication).
This proves the equation (\ref{eq1}) at an arbitrary point on $S^3$.
\end{proof}

Using a simple stereographic projection we can plot the field lines for $\B_{\rm right}$
(see \Fig{fig: b_right streamlines}).
The traced field lines are simply the Hopf rings, mutually linked circles that fill $\mathbb{R}^3$.

\begin{figure}
\begin{center}
\includegraphics[width=0.7\columnwidth]{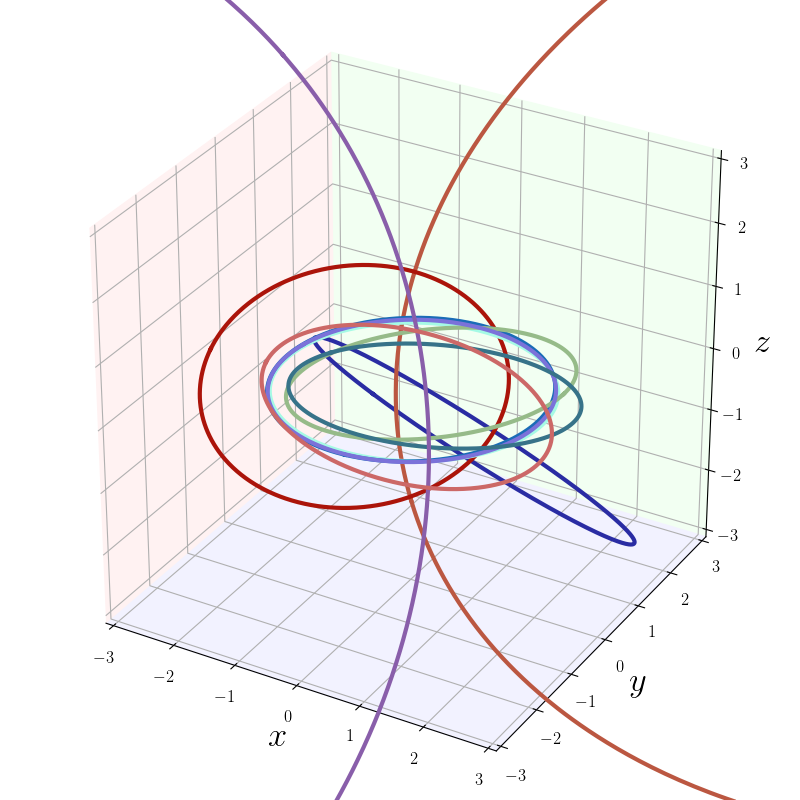} \\
\includegraphics[width=0.7\columnwidth]{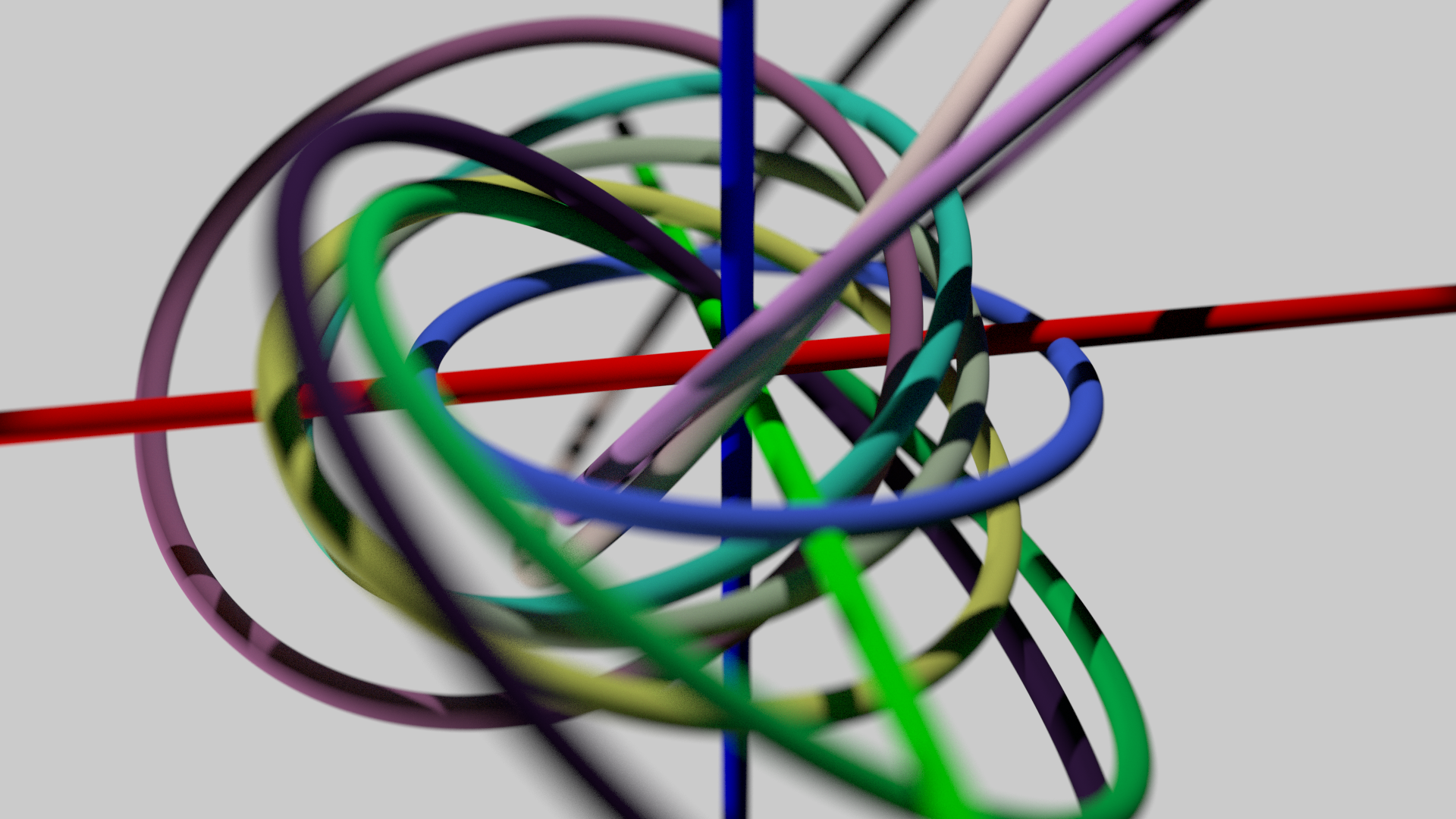}
\end{center}
\caption[]{
Two plots for the field $\B_{\rm right}$ in the stereographic projection.
Here we use a few random field lines to show case the nature of this field.
In the lower panel the axis are $x$ (red), $y$ (green) and $z$ (blue).
}
\label{fig:  b_right streamlines}
\end{figure}

It is natural to investigate the Hopf magnetic vector field from Hamiltonian dynamics.
Consider the standard symplectic form $\dd x_0 \wedge \dd x_1 + \dd x_2 \wedge \dd x_3$
in $\mathbb{R}^4$.
Consider the Hamiltonian $H(x_0,x_1,x_2,x_3)= (x_0^2 + x_1^2 + x_2^2 + x_3^2)$.
Obviously, the Hopf magnetic field $2\B_{\rm right}$ determines the Hamilton flow for $H=1$. 
In a general case a Hamilton flow is divergence-free, because it keeps the simplistic
structure and the Hamiltonian.
Thus, the fundamental $3$-form $\dd\Omega$ on the prescribed energy level $H = {\rm const}$ has
to be defined by the formula
$\dd \Omega \wedge \dd H = \dd x_0 \wedge \dd x_1 \wedge \dd x_2 \wedge \dd x_3$.
In the example the flow $\B_{\rm right}$ is integrable: the functions
$F_1(x_0, x_1, x_2, x_3) = x^2_0 + x_1^2$, $F_2(x_0, x_1, x_2, x_3) = x^2_2 + x^2_3$ are first-integrals.
After a small generic perturbation of the standard simplectic form
$\dd x_0 \wedge \dd x_1 + \dd x_2 \wedge \dd x_3 \mapsto \dd x_0 \wedge \dd x_1 + \dd x_2 \wedge \dd x_3 + \varepsilon \omega$,
where coefficients of the $2$-form $\omega$ depend on points in $\mathbb{R}^4$, we obtain a
non-integrable system with chaotic magnetic lines.
The examples from \cite{4} correspond to the non-small perturbation  
$\dd x_0 \wedge \dd x_1 + \dd x_2 \wedge \dd x_3 \mapsto \dd x_0 \wedge \dd x_1 + \dd x_2 \wedge \dd x_3 + a \dd x_0 \wedge \dd x_1$,
which admits closed knotted magnetic lines for rational parameter $a$.

\subsection{Preliminary discussion}
In the next section a notion of geodesic flows on surfaces is required.
The Hopf magnetic field represents the universal ($2$-sheets) covering over the
geodesic flow on the standard 2-sphere $S^2$.
The geodesic flow is a dynamic system on 3D manifold $\Lambda(S^2)$,
which is the spherization of the tangent bundle of $S^2$.
A point on the manifold $\Lambda(S^2)$ is pair $(x \in S^2, \vec{\xi}_x)$,
where $\vec{\xi}_x$ is a unit tangent $S^2$, attached at the point $x \in S^2$.
The manifold $\Lambda(S^2)$ is diffeomorphic to the transformation group $SO(3)$.
To proof this, it is sufficient to consider $S^2 \subset \mathbb{R}^3$.
Each point $(x, \vec{\xi}_x)$ determines a $2$-orthogonal base in $\mathbb{R}^3$
that is a point on manifold $SO(3)$.
The manifold $SO(3)$ is the base of the double covering $S^3 \to SO(3)$.

The tautological vector field
$\vec{\xi}_{x}$ on $T(\Lambda(S^2))$ determines the Hamiltonian dynamic system,
which is called the geodesic flow.
A point $x \in S^2$ moves along $\vec{\xi}_x$.
The geodesic flow, lifted on the universal covering $S^3$, determines the magnetic
(divergence-free) vector field, which is called the horizontal magnetic field on $S^3$.
The Hopf vector field is called the vertical vector field.
This vector field corresponds
to the rotation of fibers of the standard projection $\Lambda(S^2) \to S^2$. 

By the same argument one may define the geodesic field on $\Lambda(M^2)$,
where $M^2$ is a surface of a constant negative scalar curvature. 
One may take $M^2$ as the closed surface of a constant 
Riemann surface (with a constant negative scalar curvature).
However, this example is not suitable for MHD, because $\Lambda(M^2)$ admits a complicated homotopy type.

The surface $M^2$ can be non-compact and coincide with the standard Lobachevskii plane $L^2$.
In this case $\Lambda(L^2)$ also is non-compact.
One may take an isometric action $G \times L^2 \to L^2$ with locally finite orbits.
The group $G$ is called a Fuchsian group.
This action can have fixed points.
In this case the manifold $\Lambda(L^2)$ admits the quotient $\Lambda(L^2)/G$.
This quotient can be considered as an interior of a closed Riemannian 3D-manifold,
with a metric that has singularities (pinches).
For many examples $\Lambda(L^2)/G$ is a branching covering over the standard $S^3$-sphere,
equipped with the Riemanian metric, which corresponds to the standard geodesic metric
on $\Lambda(L^2)$.
We observe that in many cases of $G$ this $\Lambda(L^2)/G$ admits a Riemannian metric, which
is the conformal equivalent to the standard metric on $S^3$. 
We consider the most fundamental example of $G$, which is called the modular group.
The fundamental domain of the modular group is shown in \Fig{fig: infinit}.
There are two generators: the generator of the order 3 acts by the rotation trough the angle
$\frac{2\pi}{3}$ at the central triangle, the generator of the order 2 acts
by the central symmetry trough a point at the boundary geodesic line of the triangle.
These two generators are not commuted and the modular group is the non-commutative product $\Z_3 \ast \Z_2$.

We interpret the geodesic flow on $\Lambda(L^2)/G$ as a force-free magnetic field on
$S^3$ (Section \ref{sec: force-free config}) and as MHD-solitons (Section \ref{55}),
which are generalizations of the Hopf magnetic field.
The scalar factor of the metric we interpret as density on $S^3$, while ramifications
curves we interpret as magnetic pinches.
The example of the geodesic flow with the modular group $G$ in dynamical systems was considered by \cite{Ghy}.
In MHD this gives a testing example to calculate higher invariants of magnetic fields.  

Let us briefly explain the reason to investigate magnetic fields using this technique.
The present trend in solar physic and cosmology is to investigate the complicated fine
structure of observable magnetic fields in non-homogeneous space.
For that stability conditions for magnetic lines are required. 
The Hopf magnetic field is a force-free configuration with the global minimum of the
magnetic energy, as was discovered by \cite{ArnoldHopf1974}.
This configuration corresponds to a integrable dynamics.
After a small perturbation we get a non-integrable dynamics with complicated chaotic field lines.
For hyperbolic geodesic flows the situation is opposite.
The dynamic of the modular group itself is non-integrable, but the trajectories
are pursued in the sense of D.\ V.\ Anosov \citep{Ano67}.
After a small perturbation,
the properties of dynamics survives because the Anosov Condition,
discovered by D.\ V.\ Anosov in his famous paper
``Geodesic flows on closed Riemannian manifolds of negative curvature'' \citep{Ano67}.

\subsection{Negative Eigenfunction}
\label{sec: b_left}

The magnetic field $\B_{\rm right}$ is
generalized by the following construction.
Take $S^3$ as the unit quaternions
$\left\{a+ b\i + c\j+d\k \middle| a^2 + b^2 + c^2 + d^2 = 1\right\}$.
Take a tangent quaternion $\xi \in T_{\x=1}(S^3)$ and define the vector-field $\B_{\rm right}(\x)=\x\xi$
by the right multiplication.
In the case $\xi=\i$ we get the vector-field from Lemma \ref{l1}. 
In the case $\xi=\j$ the vector-field $\B_{\rm right}$ is not invariant with respect to the action
$\Theta$ along the Hopf fibers. To get the invariant vector-field $\B_{\rm left}$ we define
$\B_{\rm left}=\j \x$, $\x \in S^3$, by the left multiplication.
We get:
\begin{equation}\label{eq11}
\rot\B_{\rm left}(\x) = -2\B_{\rm left}(\x), \quad \x \in S^3.
\end{equation}
This follows from the fact that the conjugation
$$\left( a + b\i + c\j + d\k \right)^{*} \mapsto a - b\i - c\j - d\k,$$
which is an antiautomorphism and an isometry, transforms right vector-fields to left-vector fields. 
This antiautomorphism changes the orientation on $S^3$.
Therefore, equation \eqref{eq1} for the vector-field $\B_{\rm right}$ implies equation \eqref{eq11} for $\B_{\rm left}$.

The vector-field $\B_{\rm left}$ admits an alternative description
by means of geodesic flows on the Riemann sphere
$S^2$ in the following way.
The sphere $(S^3, g)$ is diffeomorphic to the universal (2-sheeted) covering over the manifold $SO(3)$,
equipped with the standard Riemannian metric.
The manifold $SO(3)$ is diffeomorphic to the spherization of the tangent bundle over the
standard 2-sphere $S^2$, denoted by $\Lambda(S^2)$.
The projection $p_1(\x): \Lambda(S^2) \to S^2$, $\x \in \Lambda(S^2)$ is well-defined.
A circle fiber over $p_1(\x) \in S^2, \x \in \Lambda(S^2)$
is visualized as a great circle $S^1 \subset S^2$, with the center $p_1(\x)$, equipped with the prescribed orientation.

Consider the spherization of the (trivial) tangent bundle over the plane $\Lambda(\mathbb{R}^2)$.  
Denote by $\B_{\rm left}$ the magnetic field on $\Lambda(\mathbb{R}^2)$, which is tangent to the geodesic flow.
The natural Riemannian metric $h$ on $\Lambda(\mathbb{R}^2)$ coincides with the standard metric of
the decomposition $\Lambda(\mathbb{R}^2) = \mathbb{R}^2 \times S^1$. 
 
\begin{lemma}\label{l2}
The equation:
\begin{equation}\label{eq2}
\rot \B_{\rm left}(\x) = -\B_{\rm left}(\x), \quad \x \in \Lambda(\mathbb{R}^2),
\end{equation}
in the metric $h$ is satisfied.
\end{lemma}

\begin{proof}
The manifold $\Lambda(\mathbb{R}^2)$ is equipped with the projection $p_2(\x): \Lambda(\mathbb{R}^2) \to \mathbb{R}^2$. 
Take the Cartesian coordinates in $\mathbb{R}^2$ and the coordinate $\varphi$ along  fibers.
In the coordinates $(x,y,\varphi)$ on $\Lambda(\mathbb{R}^2)$ the magnetic field $\B$ is defined as
$B_x = \cos(\varphi), B_y = \sin(\varphi), B_{\varphi}=0$.
The components of $\rot \B$ are defined by the determinant:
\begin{eqnarray}\label{rot}
\begin{array}{ccc}
\frac{\partial }{\partial x} & \frac{\partial }{\partial y} & \frac{\partial }{\partial \varphi} \\
B_x  & B_y &  B_{\varphi} \\
\ast & \ast & \ast  \\
\end{array}.
\end{eqnarray}
Lemma \ref{l2} is proven by  following calculations:
at $\x$ for $\B_{\rm left} = \B$ (one may assume $\varphi = \frac{\pi}{2}$): $B_x=0, B_y=1, B_{\varphi}=0$;
$(\rot \B)_y= -\frac{\partial B_{\varphi}}{\partial x} + \frac{\partial B_x}{\partial \varphi} = -1$, $(\rot \B)_x=(\rot \B)_z=0$.
\end{proof}

\subsection{Eigenfunctions on Different Manifolds}

Consider the spherization of the tangent bundle over the Riemannian sphere $\Lambda(S^2)$ and
the spherization of the tangent bundle over the Lobachevskii plane $\Lambda(L^2)$.
The spaces $\Lambda(S^2)$ and $\Lambda(L^2)$ are equipped with the standard Riemannian metrics $g_S$ and $g_L$.
The metrics correspond to the standard metrics on $S^2$ and $L^2$ and the standard metric on the circle.
Denote by $\B_{\rm left}$ the magnetic field on $(S^3,g)$ as the pull-back of the magnetic field on $\Lambda(S^2)$,
which is tangent to the geodesic flow.
The geodesic magnetic fields on $\Lambda(S^2)$, $\Lambda(L^2)$ are also denoted by $\B_{\rm left}$.

\begin{lemma}\label{l3}
The equation \eqref{eq2}
is satisfied 
on $(\Lambda(S^2),g_S)$ and $(\Lambda(L^2),g_L)$.
\end{lemma}

\begin{proof}
Let us prove the lemma for the space $(\Lambda(S^2),g_S)$.
For the points $\hat{\x} \in \Lambda(S^2)$ and $\hat{\y} \in \Lambda(\mathbb{R}^2)$ in the
corresponding neighborhoods $\hat{\x} \in \hat V_{\hat\x} \subset \Lambda(S^2)$,
$\hat{\y} \in U_{\hat{\y}} \subset \Lambda(\mathbb{R}^2)$, 
let us construct a mapping $\pr: \hat V_{\hat{\x}} \to \hat U_{\hat{\y}}$,  which is an isometry in 
vertical lines and is a local isometry in  horizontal planes up to $O(r^2)$, where $r$ is the distance in $U_{\hat \x}$.

Consider the natural Riemannian metric $g_{S}$ on $\Lambda(S^2)$ in $\hat V_{\hat{\x}}$ locally near a point $\hat{\x} \in \Lambda(S^2)$.
In horizontal planes the metric $g_S$ agrees with the Riemannian metric $h$ on the standard sphere $S^2 \subset \mathbb{R}^3$.
In vertical planes the metric $g_S$ corresponds to angles trough points on $S^2$.

Take a tangent plane $T_{\x} \subset \mathbb{R}^3$
at the point $\x=p_1(\hat \x) \in S^2$, where
$p_1: \Lambda(S^2) \to S^2$ is the natural projection along vertical coordinates.
Consider the stereographic projection $P$ from
$S^2_{\x}$ into $T_{\x}$, which keeps the points: $P(\x)=(\y)$, $\y= p_2(\hat \y)$, $p_2:
\Lambda(\mathbb{R}^2) \to \mathbb{R}^2$.
The projection $P$ is a conformal map and is an isometry up to $O(r^2)$ near $\x$.
This stereographic projection induces the required mapping
$\pr: \hat V_{\x} \to \hat U_{\y}$.

From equation \eqref{eq2} for $\B_{{\rm left};\mathbb{R}^2}$ on $\Lambda(\mathbb{R}^2)$
at $\y$ we get the the same equation for $P^{\ast}(\B_{{\rm left};\mathbb{R}^2})$
on $\Lambda(S^2)$ at $\x$
in the induced metric $P^{\ast}(g_S)$.
After we change the metric $P^{\ast}(g_S)$ on $\Lambda(S^2)$ into the natural metric $g_S$,
we get the same equation for $P^{\ast}(\B_{{\rm left};\mathbb{R}^2})$ at $\x$,
because the curl operator is a first-order operator.

The last required fact is the following: $P^{\ast}(\B_{{\rm left};S^2})$ in the standard
metric $g_S$ coincides with the geodesic vector-field $\B_{\rm left}$ on $\Lambda(\mathbb{R}^2)$.

To prove the lemma for 
$(\Lambda(L^2),g_L)$ we use analogous arguments: instead of the stereographic projection
$S^2 \to \mathbb{R}^2$, we take a conformal mapping by the identity $L^2 \subset \mathbb{R}^2$,
where the Lobachevskii plane $L^2$ is considered as the Poincar\`e unit disk on the Euclidean plane.
At the central point of the disk the mapping $L^2 \subset \mathbb{R}^2$ is an isometry.  
\end{proof}

\begin{remark}
Equation \eqref{eq1} corresponds with 
Lemma \ref{l3} for $\Lambda(S^2)$. The natural metric on a Hopf fiber for $\Lambda(S^2) \to S^2$ is proportional to
the natural metric of the Hopf fiber for $S^3 \to S^2$ with the coefficient $2$, because $S^3 \to \Lambda(S^2)$ is the double covering. 
\end{remark}

We now generalize the example of Lemma \ref{l3} for magnetic fields in domains with non-homogeneous density (volume-forms).
Let $(A,\x)$ be a complex neighborhood of a point $\x$, equipped with a Riemannian metric $g_A$ of a constant
negative scalar curvature surface.
In the example we get $A \subset L^2$, where $L^2$ is the Lobachevskii plane.
Let $(D,\y)$ be a complex neighborhood of a point in the Riemannian sphere $S^2$, equipped with the standard
Riemannian metric $g_D$ of a constant positive scalar curvature.

Let $f: (A, \x) \to (D, \y)$ be a conformal germ of open surfaces $A$ and $D$ with metrics $g_A$, $g_D$.
Consider the natural extension $F: (U, \hat{\x}) \to (V, \hat{\y})$
of the germ $f$, where $\hat{\x} \in U \subset \Lambda(A)$, $\hat{\y} \in V \subset \Lambda(D)$ are neighborhoods of points
$\hat{\x}, p_A(\hat{\x})=\x, p_A: \Lambda(A) \to A$,  $\hat{\y}, p_D(\hat{\y})=\y, p_D: \Lambda(D) \to D$; $U$, $V$
are equipped with the standard Riemannian metrics $g_U$ and $g_V$ correspondingly, which are defined using the metrics $g_A$ and $g_D$. 

Let us consider an extra copy of $U \subset \Lambda(L^2)$ with an exotic metric, which will be denoted by
$(\tilde U, h_{\tilde U})$.
Define in $\tilde U \subset \Lambda(L^2)$ the Riemannian metric $h_{\tilde U}$, which coincides with $g_U$ along horizontal planes $A \subset (U,\hat{\x})$
of $p_U: (U,\hat{\x}) \to (\Lambda(A),\x)$ and coincides with $\k^{-1}(\x)g_U$ along the vertical fiber of $p_U$, where $\k(\x)$
is a real positive-valued function, defined by the Jacobian $\k^2(\x)$ of $\dd f$ at $\x$ of the differential
$\dd f: (T(A),\x) \to (T(D),\y)$.

Let us consider an extra copy of $V \subset \Lambda(S^2)$ with an exotic metric, which is denoted by
$(\tilde V, h_{\tilde V})$.
Define in $\tilde V \subset \Lambda(S^2)$ the Riemannian metric $h_{\tilde V}$ that coincides with $\k^{-1}(\x=f^{-1}(\y))g_V$. 

Let $\bar V \to V$, $\bar V \subset S^3$, be the natural double covering, which is the isometry on horizontal planes and is the multiplication by $2$ in each vertical circle fibers of the standard projection $p: S^3 \to \Lambda(S^2)$.
Define in $\bar V$ a Riemannian metric $g_{\bar V}$ that coincides with $g_{V}$ along horizontal planes and with $\frac{1}{2}g_{V}$ along vertical fibers.

The Riemannian metrics $g_U$, $h_{\tilde U}$, $h_{\tilde V}$, $g_{V}$ and $g_{\bar V}$
determine the volume 3-forms $\dd U$ (the standard form in $\Lambda(L^2)$), $\dd \tilde U$,
$\dd \tilde V$, $\dd V$ (the standard form in $\Lambda(S^2)$) and $\dd \bar V$
(the standard form in $S^3$) in $U$, $\tilde U$, $\tilde V$, $V$ and $\bar V$ correspondingly.
Recall $A \subset L^2$ with the standard 2-volume form $\dd L$ on the Lobachevskii plane.
The volume form $\dd \tilde U$ is defined by $\dd \tilde U=\k(\x)\dd U$,
where $\dd U$ is the standard volume form in $U$, which is the product of the
horizontal standard 2-form $\dd L$ on
the Lobachevskii plane with the the standard vertical 1-form on the circle.
Analogously, $\dd \tilde V=\k^{-2}(\y)\dd V$, where $\dd V$ is the standard volume
form on $V=\tilde V \subset \Lambda(S^2)$.
The volume forms $\dd V$, $\dd \bar V$ coincide with the standard volume forms
($\dd V$ is the restriction of the standard volume form on $\Lambda(S^2)$, $\dd \bar V$
is the restriction of the standard volume form on $S^3$; $\dd \bar V = 2 p^{\ast}\dd V$,
where $\bar V$ is standardly identified with $V$ by $p: \bar V \to V$).
The volume forms $\dd V$, $\dd \bar V$ are equipped with the density functions
$\rho_V(\hat \y)=\k^{-2}(\y = p_V(\hat \y))$,  $\rho_{\bar V}(\bar \y)=\k^{-2}(\y=p_V \circ p(\bar \y))$. 

Let $\B_U$ be the magnetic field (horizontal) in $U$ with the metric $g_U$, which is defined by the geodesic
flows in $A$ with the metric $g_A$.
Define the magnetic field $\B_{{\rm left};\tilde U}$ in $\tilde U$ with the metric $h_{\tilde U}$ by $\B_{{\rm left};\tilde U}=\B_U$.

By construction, the metrics $h_{\tilde U}$ and $h_{\tilde V}$ agree (are isometric): $F_{\ast}(h_{\tilde U}) = h_{\tilde V}$.
Denote by $\B_{{\rm left};\tilde V}$ the magnetic field $F_{\ast}(\B_{{\rm leftU};\tilde U})$ in $\tilde V \subset \Lambda(S^2)$
with the metric $h_{\tilde V}$.
Denote by $\B_{{\rm left};V}$ the magnetic field $\k^{-3}(\hat \y)\B_{{\rm left}; \tilde V}$
in $V \subset \Lambda(S^2)$
with the standard  metric $g_{V}$ and with the variable density  $\rho_{V}(\hat \y)$. 
Denote by $\B^{S3}_{{\rm left}\bar V}$ the magnetic field $\k^{-3}(p(\bar \y))p^{\ast}(\B_{{\rm left}; \V})$
in $\bar V \subset S^3$
with the standard spherical metric $g_{\bar V}$ and with the variable density  $\rho_{\bar V}(\bar \y)$. 

\begin{lemma}\label{l4}
\begin{enumerate}
\item
In the domain $\tilde U$ the following equation is satisfied:
\begin{eqnarray}\label{eq3}
\begin{array}{cc}
\div(\B_{{\rm left};\tilde U})=0; \qquad \rot \B_{{\rm left};\tilde U}(\hat{\x}) = -\k(\x)\B_{{\rm left};\tilde U}(\hat{\x}), \\
\hat{\x} \in \tilde U, \quad \x=p_{\tilde U}(\hat{\x}) \in A,
\end{array}
\end{eqnarray}
where $\rot$ and $\div$ are defined for the Riemannian metric $h_{\tilde U}$ with the density $\rho_U(\hat \x)$.

\item
In the domain $V \subset \Lambda(S^2)$ the following equation is satisfied:
\begin{eqnarray}\label{eq4}
\begin{array}{cc}
\div(\B_{{\rm left};V}(\hat \y))=0; \qquad \rot \B_{{\rm left};V}(\hat \y) = -\B_{{\rm left};V}(\hat{\y}), \\
\hat{\y} \in V, \quad \y=p_V(\hat{\y}) \in D,
\end{array}
\end{eqnarray}
where $\rot$ is defined for the standard Riemannian metric $g_V$ with the density $\rho_V(\hat \y)$. 

\item
In the domain $\bar V \subset S^3$ the following equation is satisfied:
\begin{eqnarray}\label{eq5}
\begin{array}{cc}
\div(\B_{{\rm left};\bar V}(\bar \y))=0; \qquad \rot \B_{{\rm left};\bar V}(\hat \y) = -2\B_{{\rm left};\bar V}(\hat{\y}), \\
\bar{\y} \in \bar V, \quad \y = p_V \circ p(\bar{\y}) \in D,
\end{array}
\end{eqnarray}
where $\rot$ is defined for the standard spherical Riemannian metric $g_{\bar V}$ with the density $\rho_{\bar V}(\bar \y)$. 
\end{enumerate}
\end{lemma}

\begin{proof}
By construction, the magnetic field $\B_U$ satisfies equation \eqref{l2} in $U$.
The transformation from $U$ to $\tilde U$ is the identity, but not isometry.
The first equation \eqref{eq3} is satisfied, because the volume form in $U$ corresponds
with the metric $h_{\tilde U}$.
The transformation $\B_U \mapsto \B_{{\rm left};\tilde U}$
is frozen-in and keeps the magnetic flow.
The second equation \eqref{eq3} is satisfied, because the metric $h_{\tilde U}$ is constant
in vertical fibers and the factor $\k(\x)$ in the right side of the equation corresponds
to the partial derivatives along the vertical coordinates.
This proves equation \eqref{eq3}.

The transformation $\B_{{\rm left};\tilde U} \mapsto \B_{{\rm left};V}$ is decomposed into transformations
$$\B_{{\rm left};\tilde U} \mapsto 
\B_{{\rm left};\tilde V} \mapsto \B_{{\rm left};V}.$$
The transformation $\B_{{\rm left};\tilde U} \mapsto 
\B_{{\rm left};\tilde V} $ is an isometry
and $\B_{{\rm left};\tilde V}$ satisfies  equation \eqref{eq3} in $\tilde V$.
The transformation $\B_{{\rm left};\tilde V} \mapsto \B_{{\rm left};V}$ is conform with
the scalar factor $\k(\y)$.
This transformation changes equation \eqref{eq3} in $\tilde V$ into \eqref{eq4} in $V$ with non-uniform density.

The calculations for this transformation are as follows.
Take a domain $\tilde V$ with local coordinates $\hat{\x}=(x,y,z)$.
Take a transformation $g \mapsto \lambda g$ of the metric in $\tilde V$ into a metric in $V$
with a scale $\lambda(\hat{\x}) > 0$.
The following transformation of coordinates $x \mapsto \lambda x_1$,
$y \mapsto \lambda y_1$, $z \mapsto \lambda z_1$ is an isometric transformation
of $(\tilde{V},g)$ into $(V,\lambda g)$, where $\hat{\x}_1 =(x_1, y_1, z_1)$ are
the coordinates in $V$.
Before the transformation we get a differential $1$-form $\beta \dd z$ which is
by assumption, a proper form of the operator
$\ast \circ d$ with a proper function $-\lambda(\x)$ (see equations \eqref{analog}
with analogous calculations) in $\tilde V$. This implies $\dd(\beta\dd z) = \frac{\partial \beta}{\partial x} \dd x \wedge \dd z + \frac{\partial \beta}{\partial y} \dd y \wedge \dd z$; $\frac{\partial \beta}{\partial x} = -\lambda(\hat \x)$, $\frac{\partial \beta}{\partial y} = -\lambda(\hat \x)$. 
After the transformation we get the $1$-form $\lambda \beta \dd z_1$.
We have:
$$\dd(\lambda \beta \dd z_1) = \frac{\partial \beta}{\partial 
x} \lambda \dd x \wedge \dd z_1 + \frac{\partial \beta}{\partial y} \lambda \dd y \wedge \dd z_1 +
\frac{\beta}{\lambda} \frac{\partial \lambda}{\partial x} 
\dd x_1 \wedge \dd z_1 + \frac{\beta}{\lambda} \frac{\partial \lambda}{\partial y} \dd y_1
\wedge \dd z_1 + $$
$$\lambda\beta \dd x \wedge \frac{\partial}{\partial x}\left( \frac{1}{\lambda} \dd z\right)
+ 
\lambda\beta \dd y \wedge \frac{\partial}{\partial x}\left( \frac{1}{\lambda} \dd z\right).$$
Using $\frac{\partial}{\partial z_1} = \lambda \frac{\partial}{\partial z}$,
$\dd x \wedge \frac{\partial}{\partial x}( \frac {1}{\lambda} \dd z) =
-\frac{1}{\lambda^2}\frac{\partial \lambda}{\partial x} \dd x \wedge \dd z$, $\dd y \wedge \frac{\partial}{\partial y}( \frac {1}{\lambda} \dd z) =
-\frac{1}{\lambda^2}\frac{\partial \lambda}{\partial y} \dd y \wedge \dd z$,
we have:
$$\dd(\lambda \beta \dd z_1) = \frac{\partial \beta}{\partial x}
\lambda \dd x \wedge \dd z_1 + \frac{\partial \beta}{\partial y} \lambda \dd y \wedge \dd z_1
= \frac{\partial \beta}{\partial x} \dd x_1 \wedge \dd z_1 +
\frac{\partial \beta}{\partial y} \dd y_1 \wedge \dd z_1 = $$
$$-\lambda(\dd x_1 \wedge \dd z_1 + \dd y_1 \wedge \dd z_1).$$
This proves that $\lambda \beta \dd z_1$ is the proper $1$-form of the operator $\ast \circ \dd$
in $V$ with the proper function $-\lambda^{-1}\lambda = -1$.
Setting $\k(\x)=\lambda(\x)$, we get the required formula \eqref{eq4}.

The transformation $\B_{{\rm left};V} \mapsto \B_{{\rm left};\bar V}$ is analogous to the transformation 
$\B_U \mapsto \B_{{\rm left};\tilde U}$.
In this transformation $\B_{\rm left}$ is frozen-in and the scalar factor $2$ in the right
side of the second equation \eqref{eq5} corresponds to the transformation of the metrics
$g_{V} \mapsto g_{\bar V}$, which changes partial derivatives along the vertical coordinate.
\end{proof}

\section{Magnetic force-free configurations on non-homogeneous $S^3$}
\label{sec: force-free config}

Let $P \subset L^2$ be the right $k$-triangle (all $k$-vertices on the absolute) on the
Lobachevskii plane.
Let $f_k: P \to S^2_+$ be the conformal transformation
(the Picard analytic function in the case $k=3$) of the
square ($k$-angle) onto the upper hemisphere of the Riemannian sphere $S^2$.
The vertices $v_1,v_2, \dots,v_k$ of $P$ are mapped into
points $f(v_1),f(v_2), \dots, f(v_k)$ at the equator $S^1 \subset S^2$ and we assume that
${\rm dist}(f(v_1),f(v_2))= \dots = {\rm dist}(f(v_k),f(v_1)) = \frac{2\pi}{k}$.
Denote by $f: L^2 \to S^2$ the branched cover with ramifications at $f(v_1), f(v_2), \dots, f(v_k)$,
which is defined as the conformal periodic extension of $f_k$ on the Lobachevskii plane. It is well known that $\Lambda(S^2) = S^3/-1$, where on the right side of the formula is the quotient of the standard $3$-sphere by the antipodal involution. The fiber of $S^3 \to \Lambda(S^2) \to S^2$ over the points $f(v_1), \dots f(v_k)$ in the base is the Hopf $k$-component link, which is denoted by $l \subset S^3$. For $k=3$ link $l$ consists of 3 big circles, each two circles are linked with the coefficient $+1$,    
Denote the Jacobian of $f$ by $\k^2(\x)$, $\x \in P$, $\y = f(\x) \in S^2$.
Statement \rm{(i)} of the following lemma is a corollary from Theorem \ref{l4}. 

\begin{theorem}\label{th5}
Assume $k \ge 3$ is fixed.

\begin{enumerate}
\item
For magnetic force-free field $\B_{\rm left}$ on $S^3 \setminus l$ with the standard Riemannian
metric $g$ and the density function $\rho(\hat \y)=\k^{-2}(\y)$, $\y = p(\hat \y)$,
with the standard Hopf bundel $p: S^3 \to S^2 \to \Lambda(S^2)$,
there are $k$-component exceptional fibers $l \subset S^3$ with an infinite density.

\item
The $k$-component pinch curve $l$
of the magnetic field $\B_{\rm left}$ is the standard $k$-component Hopf link in $S^3$.
The components of $l$ are preimages of points $f(v_1), f(v_2), \dots, f(v_k)$
by the projection $p: S^3 \to \Lambda(S^2)$.

\item
In the case $k=3$ the scalar factor of the density function $\sqrt{\rho(\hat \y)}=\k^{-1}(\y)$ in equation
\eqref{eq3} has an asymptotic $(-z\ln(z))^{-1}$ near $l$, where $z$ is the distance from $\hat \y$ to $l$.
The magnetic field has the asymptotic $(-z\ln(z))^{-1}$ for $z \to 0+$.
The magnetic energy $\int B^2\ \dd\Omega$, where $B^2(\hat \y) = \k^{-2}(\hat \y)$ and
$\hat{\y} \in \Omega = S^3 \setminus l$, has the asymptotic
$\simeq \int^{+\varepsilon}_0 (z \ln^2(z))^{-1}\ \dd z < +\infty$ near a component of
a cusp curve $l$, in the standard metric on $S^3$. 

\item
In the case $k=3$, $\B_{{\rm left};3}$ is projected to tangents along trajectories of the Lorenz attractor
\cite{Ghy, 1} by a 12-sheeted branching covering $S^3 \setminus l \to S^3 \setminus l'$,
which transforms $l$ into the exceptional trefoil $l'$ of the Lorenz attractor.  
 

\item
The stereographic projection $S^3 \setminus {pt} \to \mathbb{R}^3$ transforms
$\B_{\rm left}$ into a force-free
magnetic field with a finite magnetic energy in non-homogeneous isotropic space $\mathbb{R}^3$.
This construction is analogous to \cite{3}.
\end{enumerate}
\end{theorem}

\begin{proof}$\rm{(iii)}$

Let $H^{\uparrow}$ be the upper half-plane with the complex coordinate, denoted by $w$,
$H^{\uparrow} \equiv L$, where $L$ is the Lobachevskii plane, equipped with the standard
conformal metric, $H^{\downarrow}$
be the lower half-plane, $H_+$ be the right half-plane $H_+ = \{w \in L, \vert \Re w > 0\}$
and $H_-$ be the left half-plane.
We identify $H^{\downarrow}$ with the Lobachevskii plane $L$, $H_-$ with the Riemannian half-sphere.
Let $D = \{ w \in H^{\uparrow}, \vert \tau \vert > 1, \vert \Re w \vert < 1\}$ be the triangle in $H^{\uparrow}$.
Let us consider the analytic function $F: D \to H_+$, $F(\infty)=\infty$,
$F(+1)=\i$, $F(-1)=-\i$. From the conditions we get $F(\i)=0$. 

Take the triangle $C = \{a=v_1,c=v_2,
-c=v_3\}$ on $H^{\downarrow}$, $a=0, c=+1$.
The considered triangle is mapped onto the triangle $D=\{\infty,c,-c\}$ in the upper
half-plane $H^{\uparrow}$ by $I_1: x \mapsto x^{-1}=w$
(see \Fig{fig: infinit}).

\begin{figure}
\begin{center}
\includegraphics[width=0.7\columnwidth]{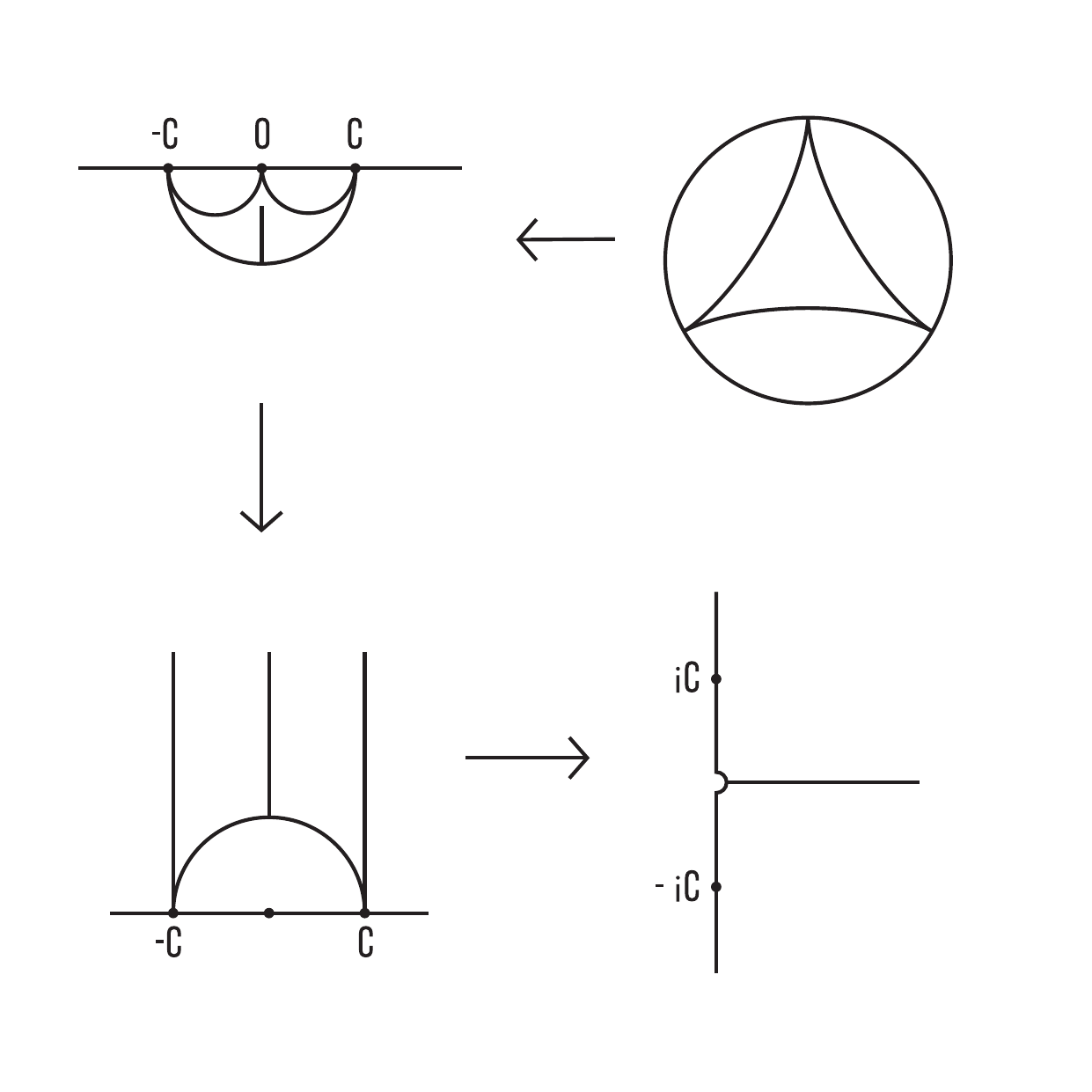}
\end{center}
\caption[]{
A transformation of the standard hyperbolic triangle onto the Riemannian half-sphere by the modular function.
}
\label{fig: infinit}
\end{figure}

The function 
$f: C \to H_-$ is the composition of the maps
$I_1: x \mapsto x^{-1}=w$, $F: D \to H^+$, $F: w \mapsto F(w)=v$, $I_2: H_+ \to H_-$, $I_2: v \mapsto v^{-1}=y$;
$f = I_2 \circ F \circ I_1: x \mapsto y$.
The function $F$ is called the modular function, this function has the asymptotic
$F \simeq w \mapsto \exp(\frac{i\pi w}{24})=v$, when $w \to +\i\infty$.
The goal is to calculate the scalar factor $\k$ near the origin $f(0)=0$ in the target domain.

In $C \subset H^{\downarrow}$ we get the metric on the hyperbolic plane, near the origin on the boundary.
The distance between two points on a vertical ray is given by the logarithmic scale.
In $H_-$ near the origin the  metric is the Euclidean metric. 

We get: $\dd y = \exp(-1/x)/x^2\dd x$ and $\frac{\dd x}{x}=\dd l$, where $l$ the distance in the domain space, $x$ is the Euclidean coordinate in the domain space, $y$ is the coordinate in the target space, which corresponds to the metric.
Therefore, the scalar factor $\k^{-1}(\hat \y)$ depends of the distance $z$ from the cusp $L$ in the target space
$S^3$ with the standard metric as follows: 
$$ \k(z) \approx -z\ln(z).$$

By this asymptotic we get the asymptotic of the magnetic energy is given by the prescribed integral over $z$. 
\end{proof}

\begin{proof}$\rm{(ii)}$, $\rm{(iv)}$

The Lorenz attractor by \cite{Ghy} coincides with the geodesic flows on the orbifold
$(2, 3, \infty)$ from \cite{1}.
The spherization of the tangent bundle over the orbifold $(2, 3, \infty)$,
which is the space of the geodesic flow,
is an open manifold diffeomorphic to the complement of the trefoil in the $3$-sphere
$S^3 \setminus l'$.
The orbifold $(2, 3, \infty)$ is the quotient 
of the Lobachevskii plane by the corresponding Fuchsian group.
The fundamental domain $P'$ of this orbifold is the triangle $\triangle O C_1 C_2$ with angles
$\left(\frac{\pi}{3},0,0\right)$.
This triangle is contained as a $\frac{1}{3}$-triangle in the triangle
$P = \triangle C_1 C_2 C_3$ with the angles $(0,0,0)$ with the vertex on the absolute
(see \Fig{fig: Circ}).
The fundamental domain $Q$ of the magnetic force-free field $\B_{\rm left}$ for
$k=3$ is the 2 sheet covering over
the space of $S^1$-fibration over the union $P \cup P_1$ of 2 triangles $P = \triangle C_1 C_2 C_3$, $P_1= \triangle C_2 C_3 C_4$, which are identified along the fibration over the common edge $(C_2 C_3)$. Therefore, the fundamental domain $Q$ is a $6$-sheeted covering space over $\Lambda(P')$.

\begin{figure}
\begin{center}
\includegraphics[width=0.7\columnwidth]{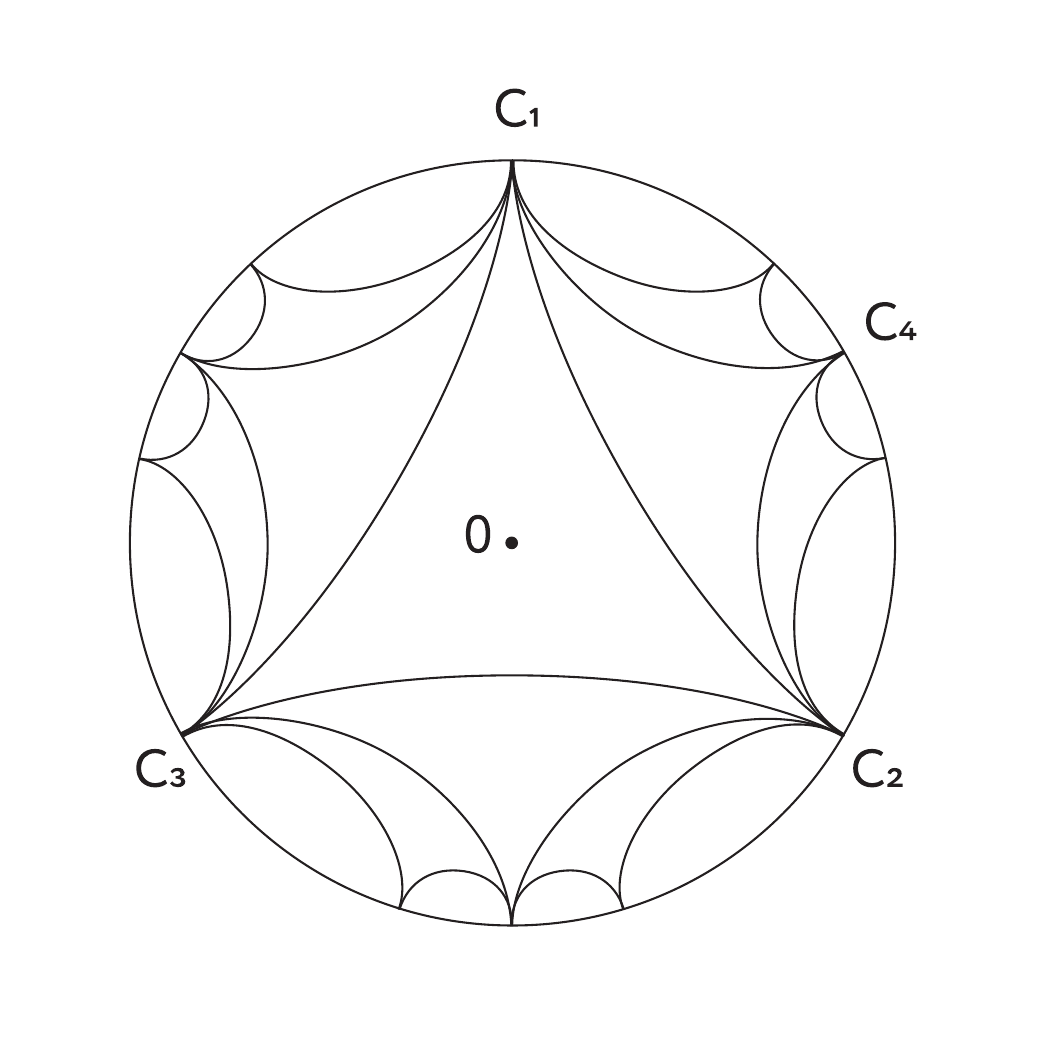}
\end{center}
\caption[]{The covering over the orbifold $(2, 3, \infty)$ for Lorenz attractor. The point $C_4$ is the image of $C_3$ with respect to the central symmetry over the fold point on $(C_1,C_2)$.}
\label{fig: Circ}
\end{figure}

According to \cite{Ghy}, the spherization of the tangent bundle $\Lambda(P')$ over the fundamental
domain $P'$ is diffeomorphic to $S^3 \setminus l'$, where $l'$ is the exceptional fiber (the trefoil),
which corresponds to the vertex of the domain $P'$, the vertex are identified by an
action of the Fuchsian group.
By the construction the spherization of the tangent bundle $\Lambda(P \cup P_1)$ over the fundamental
domain $P \cup P_1$ is diffeomorphic to $ \Lambda(S^2) \setminus l''$,
$S^3/-1 = \Lambda(S^2)$,
where $l''$ is the union of 3 exceptional fibers,
which are correspondent to the vertex $f(v_1), f(v_2), f(v_3)$ of the $\frac{1}{2}$-fundamental domain $P$.
This proves that $\Lambda(S^2) \setminus l''$ is a  $6$-covering space over $S^3 \setminus l'$,
which is branched over the trefoil $l'$.

A neighbourhood of the exceptional trefoil $l'$ in the Lorenz attractor is covered
by a non-connected neighbourhood of  $l'' \subset S^3/-1$, which is the standard 3-Hopf link.
An extra $2$-covering $S^3 \to \Lambda(S^2)$ determines the required $12$
covering $S^3 \setminus l \to S^3 \setminus l'$, which is also branched over the trefoil $l'$.
\end{proof}

\begin{remark}
By Theorem \ref{th5}, $iii$ the magnetic field $\B_{\rm left}$ on $S^3 \setminus l$
is compactified into the magnetic field on $S^3$, which tends to infinity on $l \subset S^3$.
The magnetic field $\B_{\rm left}$ is an $Z_{6}$-equivariant 
with respect to the standard action $Z_{6} \times S^3 \setminus l \to S^3 \setminus l'$ of the cyclic group
of the order $6$, therefore the magnetic field $\hat \B_{\rm left}$ on the lens quotient
$\hat Q=(S^3 \setminus l)/Z_{6}$ is well-defined.
The domain $\hat Q$ with magnetic field is a covering space over the domain with the Lorenz attractor
in $S^3$, over the exceptional fiber $l' \subset S^3$ the covering is ramified. 
\end{remark}

\section{MHD-solitons}
\label{55}

By MHD-solitons we mean quasi-stationary solutions of the ideal MHD equations.
An example of MHD-soliton was discovered by \cite{3} and generalized by \cite{4}.
For these examples the velocity field is parallel to the magnetic field.
We shall define analogous examples of solitons in $S^3$, where the velocity field is
perpendicular to the magnetic field.
Example \ref{ex: 1} is constructed by means of the Hopf magnetic field.
Example \ref{ex: 2} is a hyperbolic analog, constructed using a force-free magnetic
field from Section \ref{sec: force-free config}.
This hyperbolic soliton is defined in a non-homogeneous $S^3$ with pinches.
The velocity field is perpendicular to the gradient of the density function $\rho$.
This means that the density function depends not of time.

We consider MHD-solitons for the sphere $S^3$ with the standard metric $g$ with the constant and variable
density $\rho(\hat{\y})$, $\hat{\y} \in S^3$, see \cite{A-Kh} Remark 1.6 p.\ 262 and
Remark 1.1 p.\ 120, for the MHD-equations on a Riemannian manifold.
The density positive function $\rho(\hat{\y})$ is equivalent 
that the standard metric $g$ is changed $g \mapsto \rho(\hat{\y})^{-\frac{1}{3}} g$
by a conformal transformation.

A quasi-stationary solution means that the velocity field $\v$ does not depend on time
(see equation \eqref{eq4.2}).

\begin{eqnarray}\label{MHD}
& \frac{\partial \B}{\partial t} = -\{ \v, \B\}, \label{eq4.1}\\
& \frac{\partial \v}{\partial t} = -(\v, \nabla) \v + \rot \B \times \B - \grad p, \label{eq4.2}\\
& \div(\B)=\div(\v) = 0.
\end{eqnarray}

\begin{example}\label{ex: 1}
Assume that the standard $S^3$ is homogeneous: $\rho \equiv 1$.
Define $\v = \i$; $\B(t)=\cos(2t)\B_{\rm left} + \sin(2t)\B_{\rm left}^{\ast}$, 
where $\i$ is the vertical (right) vector field on $S^3$: the Hopf field $\i=\B_{\rm right}$,
constructed in Section \ref{sec: eigenfunctions}; $\B_{\rm left}$, $\B_{\rm left}^{\ast}$ are two horizontal (left) vector fields:
the geodesic vector field $\B_{\rm left}$, constructed in subsection \ref{sec: b_left} and its
conjugated geodesic vector field $\B^{\ast}_{\rm left}$.
By construction $\B_{\rm left}$ and $\B_{\rm left}^{\ast}$ are in the plane of the the basic
quaternion (right) vector fields $\j$, $\k$ on $S^3$.

Then, by Theorem \ref{l1}, the equation \eqref{eq4.2} is satisfied:
$\rot(\v)=2\v$, 
$\rot(\B)= - 2\B$, $(\v, \nabla) \v =0$, $\rot(\B)
\times \B = 0$.
Also equation $(\ref{eq4.1})$ is satisfied:
$-\{ \v, \B\}=\rot(\v \times \B) = -2\sin(2t)\J + 2\cos(2t)\K$. 
\end{example}

\begin{example}
\label{ex: 2}
Assume that the standard $S^3$ is non-homogeneous: $\rho(\hat \y)=\k^{-2}(\y)$, as in
Theorem \ref{th5}, $k \in 3, 4, \dots $ is fixed.
Define $\v = \rho(\hat \y)\I $, $\B(t)=\rho(\hat \y)(\cos(2t)\B_{\rm left} + \sin(2t)\B^{*}_{\rm left})$,
where $\i$ is the Hopf vertical vector field on $S^3$, $\B_{\rm left}$ is the horizontal vector field,
determined by the geodesic flow in Theorem \ref{th5}, and $\B^{*}_{\rm left}$ is vector horizontal field,
determined by the conjugated geodesic flow.
Then the equation \eqref{eq4.2} is satisfied:
$\rot(\v)=2\v$;
by Lemma \ref{l4}, equation \eqref{eq5} we get:
$\rot(\B)=-2\B$, $\rot(\B) \times \B = 0$; the equation \eqref{eq4.1} is satisfied:
$-\{ \v, \B\}= 2\rho(\hat \y)(-\sin(2t)\B_{\rm left}+ \cos(2t)\B_{\rm left}^{\ast})$. 
\end{example}

Example \ref{ex: 2} admits the following properties: structural stability and
hyperbolicity of magnetic flow.
In Example \ref{ex: 1} the Larmor radii of trajectories of particles are curved along the
direction of the velocity.
In Example \ref{ex: 2} they are curved in the opposite direction.

\section{Helicity Invariants}

Theorem \ref{l4} demonstrates
that Ghys-Dehornoy hyperbolic flows \cite{1} determines stationary solutions of MHD-equations,
which was recalled in Section \ref{55}.
As the main example we take the simplest flow with the Lorenz attractor.
We will calculate quadratic helicities for this solution.
The calculation is based on the standard arguments from ergodic theorems.
The calculation of quadratic helicities $\chi^{(2)}$ is analytic.
The calculation of $\chi^{[2]}$ is geometrical and possible with the assumption
that the magnetic field configuration admits an additional symmetry.
The calculation of $\chi^{[2]}$ for the magnetic configuration itself is an open problem. 

For a homogeneous domain $\Omega$
inequalities for magnetic field $\B$:
$$ 2\chi^{[2]} \ge \chi^{(2)}\vol^{-1}(\Omega) \ge \chi^2 \vol^{-2}(\Omega)$$
are satisfied \cite{Akhmetev-2012-278-10-ProcSteklov}.
In these inequalities $\chi^{[2]}$ and $\chi^{(2)}$ are quadratic helicities
and $\chi$ is the standard helicity.
See \cite{Akhmetev-2012-278-10-ProcSteklov} for definitions of the quadratic helicities.
All of these are invariants in ideal MHD.
For non-homogeneous domain $\Omega$
with the density function $\rho$ the inequalities are analogous (see \cite{A-C-S} the right inequality
for $\chi^{(2)}$ in a non-homogeneous domain).

For the Hopf magnetic force-free field $\B_{\rm right} = \I$ on the homogeneous $\Omega = S^3$ we get: 
$$ 2\chi^{[2]} \equiv \chi^{(2)}\vol^{-1}(S^3) \equiv \chi^2 \vol^{-2}(S^3), $$
where $\vol(S^3)$ is the volume of the sphere $S^3$.

\begin{theorem}\label{th6}
The quadratic helicity $\chi^{(2)}$ of the magnetic field $\B_{\rm left}$ in the
non-homogeneous domain $\Omega$, constructed by Theorem $\ref{th5}$,
takes the minimal possible value
$$\chi^{(2)} \equiv \frac{\chi^2}{ \vol(\Omega)},$$
where $\chi$ is the helicity of $\B_{\rm left}$.
\end{theorem} 

\begin{proof}
Let us prove that the field line helicity
function $\mathcal{A}(\x)$ \cite{Yeates-Hornig-2011-18-102118-PhysPlasm} is
constant in $\Omega = S^3 \setminus l$.  
This function is defined by the average of 
$(\A, \B)\rho$ along the magnetic line, issued from the point $\x \in \Omega$.
By equation \eqref{eq5} the vector-potential $\A$ coincides with $\frac{1}{2}\B$ and
$(\B, \B) = k^{2}(\hat \y)$,
$\rho(\hat \y) = k^{-2}(\hat \y)$ by Theorem \ref{th5} (iii).
We get the function $\mathcal{A}(\x)$ is a constant, this  implies that asymptotic linking number is uniformly
distributed in $\Omega$ and $\chi^{(2)}$ contains the minimal value.
\end{proof}

The magnetic field $\B_{\rm left}$ on $S^3$ from equation \eqref{th5} admits a cyclic $Z_{4}$-transformation $\i: S^3 \to S^3$
along the Hopf fibers, which is defined by the complex multiplication.
This transformation maps $\J$ to $-\J$ in Example \ref{ex: 1}, and maps $\B_{\rm left}$ to $-\B_{\rm left}$
in Example \ref{ex: 2}.
On the non-homogeneous domain which is the quotient $\hat \Omega = S^3/\J$ with the total volume $\vol(\hat \Omega)$
a magnetic field $\hat \B_{\rm left}$ with the prescribed local coefficient system is well-defined and the quadratic helicities
$\hat \chi^{[2]}$ and $\hat \chi^{(2)}$ are well-defined.
This construction is motivated by \cite{Ze} as a model of superconductivity.

\begin{theorem}\label{th7}
The quadratic helicities $\hat \chi^{[2]}$ and $\hat \chi^{(2)}$, and the helicity $\chi$ of $\hat \B_{\rm left}$
in $\hat \Omega$ satisfy the equation:
$$ \hat \chi^{[2]} \equiv 2 \hat \chi^{(2)} \vol^{-1}(\hat \Omega) \equiv \hat \chi^{2}\vol^{-2}(\hat \Omega).$$
\end{theorem}

\begin{proof}
Let us calculate quadratic helicities for magnetic field in $S^3/\i =\Lambda(S^2)/I$, equipped with the metric on the Lobachevskii plane $L$. 

Take the universal branching covering
$L \times S^1 \to \Lambda(S^2)/I$ which is the quotient of the covering space $L \times S^1$ by
the corresponding Fuchsian group $G$.
A magnetic line $l$ in $\Lambda(S^2)/I$ is represented
by the corresponding collection $\{\lambda_i\}$ of non-orientable geodesics on the Poincar\'e plane, invariant with respect to $G$.
For rational geodesic the collection 
$\{\lambda_i\}$ is finite in the fundamental domain $P \subset L$ of $G$.
For generic $l$ the collection $\{\lambda_i\}$ is dense in $L$.
Because the involution $I: \Lambda(S^2) \to \Lambda(S^2)$, geodesics $\lambda_i$ and
$-\lambda_i$ with the opposite orientation are correspondingly identified.

The linking number $n(l_1 \cup I(l_1),l_2 \cup I(l_2))$ between two closed magnetic lines $l_1$, $l_2$
is calculated as number of intersection points in the fundamental domain $P$ of the two
collections $\{\lambda_{1,j}\}$, $\{\lambda_{2,i}\}$ of rational geodesics.
Each intersection point is taken with the negative sign.
This statement is a particular case of a Birkhoff's Theorem about linking number of two acyclic geodesics.
The collection $\{\lambda_i \cup -\lambda_i\}$ is acyclic (is null-homologous). 
A calculation of the linking number $n(l_1,l_2)$ is complicated \cite{1}.

Denote $l_a \cup I(l_a)$, $a=1,2$ by $\bar{l}_a$.
After the normalization of the linking number with respect to magnetic lengths
of $\bar{l}_1,\bar{l}_2$, we get much simpler calculation of $n(\bar{l}_1,\bar{l}_2)$.
The number of intersection points in $P$ of two geodesic is calculated as $\tau^2S^{-1}(P)$,
where $\tau$ is the natural parameter on geodesic, $S(P)$ is the square of the domain $P$ (the complete proof is based on ergodicity and is omitted).
We get $\tau^{-2}n(\bar{l}_1,\bar{l}_2) = (\pi(k-2))^{-1}$, where $\tau$ is the parameter
of the magnetic lengths, $\pi(k-2)$ is the square of the fundamental domain ($k$-angles) $P(k)$ on the Lobachevskii plane. 

For the square of the helicity we get:
$$ \hat{\chi}^2 = (\pi(k-2))^{-2} \vol^4(\Lambda(S^2)/I).$$
For the quadratic helicity $\hat{\chi}^{(2)}$ is better to use the formula for triples magnetic lines,
(see \cite{A-C-S} and \cite{Akhmetev-2012-278-10-ProcSteklov}).
We get:
$$ \hat{\chi}^{(2)} = (\pi(k-2))^{-2} \vol^3(\Lambda(S^2)/I).$$
For the quadratic helicity $\hat{\chi}^{[2]}$ we get:
$$ 2\hat{\chi}^{[2]} = (\pi(k-2))^{-2} \vol^2(\Lambda(S^2)/I).$$
\end{proof}

\[  \]
Petr Akhmet'ev was supported in part by the RFBR GFEN 17-52-53203;
RFBR 16-51-150005.
Simon Candelaresi acknowledges financial support from the UK's
STFC (grant number ST/K000993).
Alexandr Smirnov acknowledges financial support from the Ministry of Education and Science of the Russian Federation in the framework of increase Competitiveness Program of NUST "MISIS".

\bibliographystyle{jpp}
\bibliography{references}

\end{document}